\documentclass[12pt]{article}

\usepackage{amsmath}
\usepackage{amsfonts,amssymb,amsthm,overpic}
\makeatletter
\newcommand{\diam}{\mathop{\operator@font diam}}
\usepackage{setspace}
\usepackage{multicol}
\usepackage{multirow}
\usepackage[compact]{titlesec}

\usepackage{epsfig}

\newtheorem{definition}{Definition}[section]

\newtheorem{theorem}{Theorem}[section]
\newtheorem{proposition}{Proposition}[section]

\newtheorem{lemma}{Lemma}[section]

\newcommand{\cT}{\mathcal{T}}

\begin{document}

\title{\Huge{\textsc{The Order on the Light Cone and its induced Topology}}}

\author{Kyriakos Papadopoulos$^1$, Santanu Acharjee$^2$, Basil K. Papadopoulos$^3$\\
\small{1. Department of Mathematics, Kuwait University, PO Box 5969, Safat 13060, Kuwait}\\
\small{2. Department of Mathematics, Debraj Roy College, India}\\
\small{3. Department of Civil Engineering, Democritus University of Thrace, Greece}}

\date{}

\maketitle

\begin{abstract}
In this article we first correct a recent misconception about a topology that was suggested by Zeeman as a possible alternative to his Fine topology. This misconception appeared while trying to establish the causality in the ambient boundary-ambient space cosmological model. We then show that this topology is actually the intersection topology (in the sense of G.M. Reed) between the Euclidean topology on $\mathbb{R}^4$ and the order topology whose order, namely horismos, is defined on the light cone and we show that the order topology from horismos belongs to the class of Zeeman topologies. These results accelerate the need for a deeper and more systematic study of the global topological properties of spacetime manifolds.
\end{abstract}


\section*{Acknowledgement} The Authors would like to thank the Reviewers for the creative
comments and remarks, as well as for their corrections, which have improved the text
significantly.

\section{Introduction.}

\subsection{Causality.}

In spacetime geometry, one can introduce three causal relations, namely, the
chronological order $\ll$, the causal order $\prec$ and the relation horismos $\rightarrow$ (that we will call an ``order'' conventionally), and these can be meaningfully extended to  any \emph{event-set}, a set
$(X,\ll,\prec,\rightarrow)$ equipped with all three of these orders having no metric \cite{Penrose-Kronheimer,Penrose-difftopology}. In this context we say that the event $x$ chronologically precedes an event $y$, written
$x\ll y$ if $y$ lies inside the future null cone of $x$, $x$ {\em causally precedes} $y$,
$x \prec y$, if $y$ lies inside or on the future null cone of $x$ and  $x$ is at {\em horismos}
with $y$, written $x \rightarrow y$, if $y$ lies on the future null cone of $x$. The chronological order is irreflexive, while the causal order and horismos are reflexive. Then, the notations $I^+(x) = \{y \in M : x \ll y\},  J^+(x) = \{y \in M : x \prec y\}$ will be used for the chronological and the causal futures  of $x$ respectively (and with a minus instead of a plus sign  for the pasts), while the future null cone of $x$ will be denoted by $\mathcal{N}^+(x)\equiv\partial J^{+}(x)= \{y \in M : x \rightarrow y\}$, and dually for the null past of $x$, cf. \cite{Penrose-difftopology}. Our metric signature is timelike, $(+,-,-,-)$. $Q$ denotes the characteristic quadratic form on $M$, given by
$Q = x_0^2-x_1^2-x_2^2-x_3^2$, $x= (x_0,x_1,x_2,x_3) \in M$
and $<$ is the partial order on $M$ given by $x<y$ if the
vector $y-x$ is timelike and pointing towards the future (in this sense, $<$ is actually the chronological order $\ll$).

The above definitions of  futures and pasts of a set can be trivially extended to the situation of any partially ordered set $(X,<)$.
In a purely topological context this is usually done by passing to the so-called upper (i.e. future) and lower (i.e. past)
sets which in turn lead to the  future  and past  topologies (see \cite{Compendium}).
A subset $A \subset X$ is a {\em past set} if $A = I^-(A)$ and dually for the future. Then, the {\em future topology} $\cT^+$ is generated
by the subbase $\mathcal{S}^+ = \{X \setminus I^-(x) : x \in X\}$
and the {\em past topology} $\cT^-$ by $\mathcal{S}^- = \{X \setminus I^+(x)  : x \in X\}$.
The {\em interval topology} $\cT_{in}$ on $X$ then consists of basic sets which are finite intersections
of subbasic sets of the past and the future topologies. This is in fact the topology that fully characterizes a given order of the poset $X$. Here we clarify that the names ``future topology'' and ``past topology'' are the best possible inspirations for names that
came in the mind of the authors, but are not standard in the literature. The motivation was
that they are generated by
complements of past and future sets, respectively (i.e. closures of future and
past sets, respectively). Also, the authors wish to highlight the distinction between
the interval topology $\cT_{in}$ which appears to be of an important significance in
lattice theory, from the ``interval topology'' of A.P. Alexandrov (see \cite{Penrose-difftopology}, page 29).
$\cT_{in}$ is of a more general nature, and it can be defined via any relation, while the Alexandrov
topology is restricted to the chronological order. These two topologies are different in nature,
as well as in definition, so we propose the use of ``interval topology'' for $\cT_{in}$ exclusively,
and not for the Alexandrov topology. It is worth mentioning that the Alexandrov topology being
Hausforff is equivalent to the Alexandrov topology being equal to the manifold topology which is equivalent
to the spacetime being strongly causal; the Zeeman topologies that we mention is this paper are
finer than the manifold topology.

The so-called \emph{orderability  problem} is concerned with the conditions under which the topology $\cT_<$ induced
by the order $<$  is equal to some given topology $T$ on $X$ (\cite{Good-Papadopoulos} and \cite{Orderability-Theorem}).

In this paper we correct a misconception about Zeeman $Z$ topology that appeared in \cite{Ordr-Ambient-Boundary} and we reveal the nature of the actual order interval topology which is induced from the order horismos, on the light cone.

\subsection{The Class of Zeeman Topologies.}

The class $\mathfrak{Z}$, of Zeeman topologies, is the class of topologies
on $M$ strictly finer than the Euclidean topology,
which have the property that they induce the $1$-dimensional Euclidean topology on every time
axis and the $3$-dimensional Euclidean topology on every space
axis.

Zeeman (see \cite{Zeeman2} and \cite{Zeeman1}) showed that
the causal structure of the light cones on the Minkowski
space determines its linear structure. After initiating
the question on whether a topology on Minkowski space, which
depends on the light cones, implies its linear structure as well,
he constructed the Fine Topology (that we will call Zeeman $F$
topology) which is defined as the finest topology for $M$ in the class
$\mathfrak{Z}$.

$F$ satisfies, among other properties, the
following two theorems:

\begin{theorem}\label{Theorem2}
Let $f : I \to M$ be a continuous map of the unit interval $I$
 into $M$. If $f$ is strictly order-preserving, then the image $f(I)$ is a
 piecewise linear path, consisting of a finite number of intervals along time axes.
 \end{theorem}

Let $G$ be the group of automorphisms of $M$, given by the Lorentz group,
translations and dilatations.

\begin{theorem}\label{main}
The group of homeomorphisms of the Minkowski space under $F$ is $G$.
\end{theorem}

Goebel (see \cite{gobel}) showed that the results of Zeeman are valid
without any restrictions on the spacetime, showing in particular that
the group of homeomorphisms of a spacetime $S$, with respect to the
general relativistic analogue of $F$, is the group of all homothetic
transformations of $S$.

Having Goebel's paper \cite{gobel} in mind, from now on we will denote
any spacetime with the letter $M$, without particularly restricting
ourselves
to Minkowski spacetime.

Goebel defined Zeeman topologies in curved spacetimes as follows.
Let $E^4$ be the four-dimensional Euclidean topology, let $M$ be
a spacetime manifold and let $S$ be a set of subsets of $M$. A set $A \subset M$
is open in $Z(S,E^4)$ (a topology in class $\mathfrak{Z}$), if $A \cap B$ is open
in $(B,E_B^4)$ (the subspace topology of the natural manifold topology $(M,E^4)$
with respect to $(B,E^4)$), for all $B \in S$.

\subsection{Notation}

Throughout the text, $B_\epsilon^E(x)$ will denote the Euclidean ball around $x$, radius $\epsilon$. $C^T(x) = \{y: y=x \textrm{ or } Q(y-x)>0\} $ will denote the time cone of $x$, $C^L(x) = \{y : Q(y-x)=0\}$ the light cone of $x$ and $C^S(x)= \{y: y=x \textrm{ or } Q(y-x)<0\}$ the space cone of $x$. Last, but not least, for time and space axes we will adopt Zeeman's notation $g\mathbb{R}$ and $g\mathbb{R}^3$, respectively (see \cite{Zeeman1}).

\section{The Zeeman $Z$ Topology revisited.}

We call $Z$ the topology that is mentioned by Zeeman
(see \cite{Zeeman1}) as an alternative topology for $F$. This topology is coarser than the Fine
Zeeman topology $F$, and it has a countable base of open sets of
the form:
\[ Z_\epsilon (x) = B_\epsilon^E(x) \cap (C^T(x) \cup C^S(x)) \]

The sets $Z_{\epsilon}(x)$ are open in $M^F$ (the manifold $M$ equipped with $F$) but not in
$M^E$ (the manifold $M$ equipped with the Euclidean topology $E$). In addition,
the topology $Z$ is finer than $E$.  Theorem \ref{main} is satisfied, among
other properties that $F$ satisfies as well, but Theorem \ref{Theorem2}
is not satisfied. According to Zeeman, $Z$ is technically simpler than $F$,
but it is intuitively less attractive than $F$.

In \cite{Ordr-Ambient-Boundary} we have claimed the following theorem:

\begin{theorem}\label{3}
The order horismos $\rightarrow$ induces the Zeeman topology $Z$, on $M$.
\end{theorem}

In this article, we observe that the above conjecture has error.
But, before we analyse this, let us see which is the real interval topology
that is induced by the order horismos $\rightarrow$.



\section{The Order Horismos and its induced Topology.}

In \cite{Ordr-Ambient-Boundary} the authors claimed (Theorem 2.1) that the order horismos $\rightarrow$ induces the Zeeman $Z$ topology on $M$. Unfortunately, this is due to a miscalculation with respect to the complements of the sets  $\mathcal{N}^{\pm}(x)$ (equation $(1)$ of \cite{Ordr-Ambient-Boundary}). In a backward analysis, starting from the Zeeman $Z$ basic open set $Z(x) = (B_\epsilon^E(x) - C^T(x))\cup \{x\}$, about an event $x \in M$, one can see that the subbasic open sets of the upper (future) and lower (past) topology cannot be complements of bounded sets. In particular, they are not complements of the future light cone at an event $x$ intersected with the Euclidean ball at $x$ radius $\epsilon$. This remark brings a disappointment that the Zeeman $Z$ topology is not an order topology (and it is certainly not the causal order on the Ambient Boundary, as suggested in \cite{Ordr-Ambient-Boundary}!), but it is still linked with the original order topology which is induced by the horismos $\rightarrow$, but in the case where horismos is irreflexive. Before we discuss about this link, we can easily see the following proposition holds.

\begin{proposition}\label{real horismos topology}
The interval topology which is induced by the irreflextive horismos $\rightarrow$, denoted by $T_{in}^{\rightarrow}$ has basic open sets which can be described as follows: for an event $x \in M$, a basic open set $A(x)$ is obtained by subtracting the null cone of $x$ and replacing $x$ back, i.e. $A(x) = (M - L(x)) \cup \{x\}$.
\end{proposition}

Obviously, the open ``balls'' $A(x)$ are unbounded. Perhaps, this is related with the nature of light and
the definition of the light cone of $x$, where the distance between two massless objects is actually negligible. What is sure
though is that this topology is a Zeeman topology, as well!

\begin{theorem}
The topology $T_{in}^{\rightarrow}$ belongs to the class $\mathfrak{Z}$ of Zeeman topologies.
\end{theorem}
\begin{proof}

Let $x \in A \cap B$, where $A \in T_{in}^{\rightarrow}$ and $B \in S$,
where $S$ is a set of subsets of $M$. Obviously, $B$ is open
in the subspace topology $(B,E_B^4)$, of the natural manifold
topology $(M,E^4)$. Then, $x \in A$ and
$x \in B$, which implies $x\in (M - C^L(x))\cup \{x\} \subset A$ and
$x\in B_\epsilon(x) \subset B$. The latter implies $x \in [(M-C^L(x)) \cup \{x\}]\cap B_\epsilon(x)
\subset A \cap B$, which implies $x \in [B_\epsilon(x) - C^L(x) \cup \{x\}] \subset A \cap B$.
So, $A \cap B$ is open in $(B,E_B^4)$.
\end{proof}

Our disappointment that the Zeeman topology $Z$ is not an order topology is replaced
by the observation that it is actually an intersection topology (see G.M. Reed's paper \cite{Intersection})
and, in particular, the intersection topology between $T_{in}^{\rightarrow}$ and the Euclidean
topology $E$, on $\mathbb{R}^4$. Before we proceed, we have to solve some technical issues first.

\begin{definition}\label{intersection topology}
If $T_1$ and $T_2$ are two distinct topologies on a set $X$, then the {\em intersection topology} $T^{int}$
with respect to $T_1$ and $T_2$, is the topology on $X$ such that the set $\{U_1 \cap U_2 : U_1 \in T_1, U_2 \in T_2\}$
forms a base for $(X,T)$.
\end{definition}

The condition of Definition \ref{intersection topology} is quite restrictive, and
it does not seem easy to discover if a topology satisfies it.

Here we add the following lemma, which will help us to prove our main theorem.

\begin{lemma}\label{before main theorem}

Let $T_1$ and $T_2$ be two topologies on a set $X$, with bases $\mathcal{B}_1$ and $\mathcal{B}_2$ respectively and
let $T^{int}$ be their intersection topology, provided that it exists. Then,
the following two hold.

\begin{enumerate}

\item The collection $\mathcal{B}^{int} = \{B_1 \cap B_2 : B_1 \in \mathcal{B}_1, B_2 \in \mathcal{B}_2\}$ forms
a base for $T^{int}$.

\item If $\mathcal{B}^{int}$ is a base for a topology, then this topology is $T^{int}$.

\end{enumerate}
\end{lemma}
\begin{proof}
\begin{enumerate}

\item Let $G \in T^{int}$. Then, $G= U_1 \cap U_2 \in T^{int}$, where $U_1 \in T_1$ and $U_2 \in T_2$, therefore
$\mathcal{B}_1$ is a base for $T_1$ and $\mathcal{B}_2$ is a base for $T_2$. So, $U_1 = \cup B_1$, $U_2 = \cup B_2$. Moreover,
$(\cup B_1)\cap (\cup B_2) = U_1 \cap U_2$, which implies $\cup (B_1 \cap B_2) = U_1 \cap U_2 \in T^{int}$. Thus, $\mathcal{B}^{int}$ is a base for $T^{int}$.

\item Asserting that $\mathcal{B}^{int}$ is a base for a topology $T$ different from $T^{int}$ leads trivially
to a contradiction, by the fact that $\cup(B_1 \cap B_2) = (\cup B_1)\cap (\cup B_2) = U_1 \cap U_2$.

\end{enumerate}

\end{proof}

In view of Lemma 3.1, we formulate the following theorem.

\begin{theorem}\label{Zeeman Z is intersection}
The Zeeman $Z$ topology is the intersection topology between the Euclidean topology $E$ on $\mathbb{R}^4$ and $T_{in}^{\rightarrow}$. Consequently, every open set in $Z$ is the intersection between an open set in $E$ and an open set in $T_{in}^{\rightarrow}$.
\end{theorem}

\section{Discussion.}

A question that one might bring is why should one give more credit to the interval topology $T_{in}^{\rightarrow}$ and not to the order topology $T_{\rightarrow}$ which is induced by the order in a straighforward way, giving light cones as open sets. An answer can be achieved by looking at the example of the real line $\mathbb{R}$.

The reflexive order $\leq$, on the set of real numbers $\mathbb{R}$, induces
an interval topology $T_{in}^{\leq}$ which equals to the natural topology
on this set, where the basic open sets are intervals of the form $(a,b)$. This
natural topology is induced in an immediate way from the irreflexive order $<$,
since the order $\leq$ would give open sets of the form $[a,b]$. So, in $\mathbb{R}$,
we have that $T_{in}^{\leq} = T_<$. This does not give a hint for any generalisation,
and we can see this if we compare $T_{in}^{\rightarrow}$ which gives open balls
$B_1(x) = (M - C^L(x))\cup \{x\}$ and $T_{\rightarrow}$ whose open sets are the
light cones.

Furthermore, if we choose between $T_{in}^{\rightarrow}$ and $T_{\rightarrow}$, the
interval topology, as a lattice topology, it gives more information about the spacetime (at least for the
order $\rightarrow$) than $T_{\rightarrow}$, where the second topology is more closely
related to girders and links (see \cite{Penrose-Kronheimer}).

An analytical proof of Theorem \ref{Theorem2}, for the topology $T_{in}^{\rightarrow}$, would
certainly reveal more of the properties of this topology; it certainly does not look an easy
task.

In \cite{Singularities on Amb B} a misuse of topologies in the class
$\mathfrak{Z}$ leaded into a wrong conclusion about the convergence of causal curves in \cite{Topology-Ambient-Boundary-Convergence}.
One more misconception about the nature of the topology $Z \in \mathfrak{Z}$ lead
us in a wrong conclusion about the causality in the ambient boundary (see \cite{Ordr-Ambient-Boundary}).
In conclusion, the whole idea of the Ambient Space-Ambient Boundary model should be
reviewed from its basics, and this time topology should be taken seriously and into
account, from the very first step of the construction. Otherwise, one would have to agree
with Roger Penrose's statement in \cite{Road-to-reality} that ``...higherdimensionality for spacetime has...no more complelling than of a cute idea...''.

A rigorous study of the class $\mathfrak{Z}$ and a classification of Zeeman topologies
with respect to their properties will give an answer to the challenge that we have set
in \cite{Singularities on Amb B} as well as in \cite{Ordr-Ambient-Boundary} (here
we restate this question in a more integrated form): describe
the evolution of a spacetime with respect to the class $\mathfrak{Z}$, so that one
assigns appropriate topologies of this class locally
as well as globally. It seems, for example, that the interval topology from horismos
$\rightarrow$ could give a sufficient description of the planck time and objects like black holes.

\end{document}